\documentclass[11pt]{iopart}
\usepackage{iopams} 
\usepackage{graphicx}
\usepackage{xcolor}
\usepackage{url}
\usepackage{times}
\usepackage{bbm}
\usepackage{amssymb,amsthm,amsfonts,amstext,color}
\usepackage[english]{babel}
\usepackage[colorlinks=true,breaklinks=true]{hyperref}
\usepackage{paralist} 
\usepackage{bbold}
\usepackage{float}
\usepackage{multicol}
\usepackage{etoolbox}

\newcommand{\corr}{\mathrm{corr}}

\newcommand\av[1]{\langle #1 \rangle}
\newcommand\dd{{\mathrm{d}}}

\newcommand{\cov}{{\rm cov}} 
\newcommand{\boundary}{I}

\newcommand{\bra}[1]{\langle #1|}
\newcommand{\ket}[1]{|#1\rangle}

\newcommand{\proj}[1]{\vert #1\rangle\!\langle#1 \vert}
\newcommand{\norm}[1]{\| #1 \|}

\newcommand{\norminf}[1]{\| #1 \|_\infty}

\newtheorem{lemma}{Lemma}

\makeatletter
\def\@mkboth#1#2{}
\newlength\appendixwidth
\preto\appendix{\addtocontents{toc}{\protect\patchl@section}}
\newcommand{\patchl@section}{%
  \settowidth{\appendixwidth}{\textbf{Appendix }}%
  \addtolength{\appendixwidth}{1.5em}%
  \patchcmd{\l@section}{1.5em}{\appendixwidth}{}{\ddt}%
}

\makeatother

\begin{document}

\title{
Locality of temperature in spin chains
}  

\author{Senaida Hern\'andez-Santana$^1$, Arnau Riera$^1$, Karen V. Hovhannisyan$^1$, Mart\'i Perarnau-Llobet$^1$, Luca Tagliacozzo$^1$ and Antonio Ac\'in$^{1,2}$} 

\address{$^1$ ICFO-Institut de Ciencies Fotoniques, Mediterranean Technology Park, 08860 Castelldefels (Barcelona), Spain}

\address{$^2$ ICREA-Instituci\'o Catalana de Recerca i Estudis Avan\c{c}ats, Lluis Companys 23, 08010 Barcelona, Spain}

\begin{abstract}
In traditional thermodynamics, temperature is a local quantity: a subsystem of a large thermal system is in a thermal state at the same temperature as the original system. For strongly interacting systems, however, the locality of temperature breaks down. We study the possibility of associating an effective thermal state to subsystems of infinite chains of interacting spin particles of arbitrary finite dimension. We study the effect of correlations and criticality in the definition of this effective thermal state and discuss the possible implications for the classical simulation of thermal quantum systems.
\end{abstract}

\maketitle

\tableofcontents

\section{Introduction}

The question whether the standard notions of thermodynamics are still applicable in the quantum regime has experienced a renewed interest in the recent years. This refreshed motivation can be explained as the consequence of two successes. On the one hand, the spectacular progress of the experiments encompassed in the so called \emph{quantum simulators} already allows for a direct observation of thermodynamic phenomena in many different quantum systems, such as ultra cold atoms in optical lattices, ion traps, superconductor qubits, etc \cite{bloch,polkan,plenio,greiner}. On the other hand, the inflow of ideas from quantum information theory provided significant insight into the thermodynamics of quantum systems \cite{vedral,popka,oppen,lidia}. Specifically, qualitative improvements have been made in understanding how the methods of statistical mechanics can be justified from quantum mechanics as its underlying theory \cite{polkan,popka,tumulka,Gemmer,short,brandao}.

One of the fundamental postulates of thermodynamics is the so called Zeroth Law: two bodies, each in thermodynamic equilibrium with a third system, are in equilibrium with each other \cite{reichl,balian}. This is the law that stands behind the notion of temperature \cite{reichl,balian}. In fact, the above formulation of the Zeroth law consists of three parts: (i) there exists a thermal equilibrium state which is characterized by a single parameter called temperature, and isolated systems tend to this state \cite{polkan,popka,tumulka,Gemmer}; (ii) the temperature is local, namely, each part of the whole is in a thermal state \cite{reichl,balian}; and (iii) the temperature is an intensive quantity: if the whole is in equilibrium, all the parts have the same temperature \cite{reichl,balian,ferraro,artur,Kliesch:2014,pachon}. The last two points are usually derived from statistical mechanics under the assumption of weakly interacting systems. Nevertheless, when the interactions present in the system are non-negligible, the points (ii) and (iii) need to be revised. Following the direction given by Refs.~\cite{ferraro,artur}, in this work we concentrate on the clarification and generalization of the aforementioned aspects of the Zeroth Law of the thermodynamics for spin chains with strong, short range interactions. 

The general setting of the problem is as follows. The system with Hamiltonian $H$ is in thermal equilibrium, described by a canonical state at inverse temperature $\beta$,
\begin{equation}\label{eq:thermal-state}
\omega(H) =\frac{ \e^{-\beta H} }{ \Tr\left(\e^{-\beta H}\right)} \,,
\end{equation}
and we seek to understand the thermal properties of a finite part of the system.

Obviously, in the presence of strong interactions, the reduced density matrix of a subsystem of the global system (especially in the quantum regime \cite{pachon}) will not generally have the same form as (\ref{eq:thermal-state}). In lattice systems, where the Hamiltonian is a sum of local terms interacting according to some underlying graph, it is unclear how one can  locally assign temperature to a subsystem. More precisely, the reduced density matrix of the subsystem $A$ (see Fig.~\ref{fig:BoundaryRegion}) of a global thermal state is described by
\begin{equation}
\rho_A=\Tr_{\bar A}(\omega(H))\, ,
\end{equation}
which will not be thermal unless the particles in $A$ do not interact with its environment $\bar A$. Hence, given only a subsystem state $\rho_A$ and its Hamiltonian $H_A$, it is not possible to assign a temperature to it, since this would totally depend on the features of the environment and the interactions that couple the subsystem to it. 

In the context of quantum information, a first step to circumvent the problem of assigning temperature to a subsystem was made in Ref.~\cite{ferraro}. There, for harmonic lattices it was shown that it is sufficient to extend the subsystem $A$ by a boundary region $B$ that, when traced out, disregards the correlations and the boundary effects (see Fig.~\ref{fig:BoundaryRegion}). If the size of such a boundary region is independent of the total system size, temperature can still be said to be \emph{local}.

More explicitly, given a lattice Hamiltonian $H$ with a subsystem $A$, a shell region around it $B$ and its environment $C=(A\cup B)^c$, see Fig.~\ref{fig:BoundaryRegion}, we aim to understand how the expectation values of operators that act non-trivially only on $A$ for the global thermal state $\omega(H)$ differ from those taken for the thermal state of the truncated Hamiltonian $H_{AB}$ (with $AB=A \cup B$) as the width $\ell_B$ of the boundary region $B$ increases:
\begin{equation} \label{dist1}
\Tr\bigl[O\, \rho_A'\bigr] - \Tr\bigl[O\,\rho_A\bigr]\le ||O||_\infty f(\ell_B)\, ,
\end{equation}
where $\rho_A'=\Tr_B\omega(H_{AB})$ is the state of $A$ for the chain truncated to $AB$, and $f(\ell_B)$ is expected to be a monotonically decreasing function in $\ell_B$. The width of the boundary region $\ell_B$ is defined as the graph-distance between the sets of vertices (regions) $A$ and $C$.

Surely, the differences (\ref{dist1}) fully characterize the distance of $\rho'_A$ from $\rho_A$. Indeed, the trace distance, a distinguishability measure for quantum states, $D(\rho'_A,\rho_A)=\frac{1}{2}||\rho'_A-\rho_A||_1$, has the following representation \cite{mikeike,FvdG}:
\begin{equation} \label{dist2}
D(\rho'_A,\rho_A)=\max_{0<O<I}\Tr\bigl[ O(\rho'_A-\rho_A) \bigr]\leq f(\ell_B),
\end{equation}
where $I$ is the identity operator in the Hilbert space of $A$.

\begin{figure}
\begin{center}
\includegraphics[scale=.65]{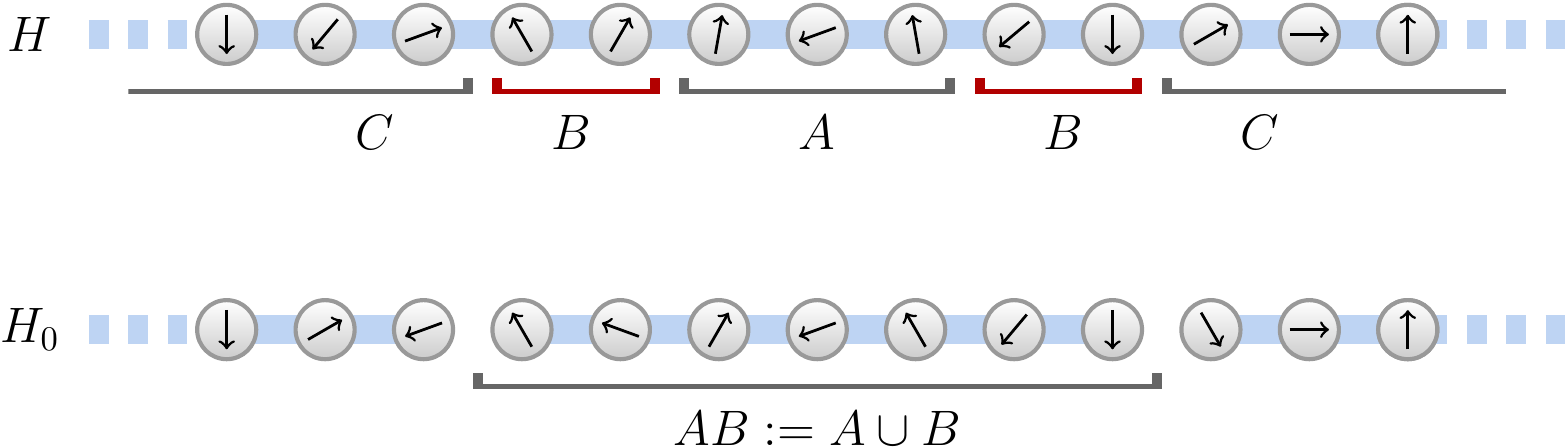}
\end{center}
\caption{
Scheme of the subsystem of interest $A$, the boundary region $B$, 
and their environment $C$ for a spin chain. Expectation values on 
$A$ for the thermal state of the full Hamiltonian $H$ (above) are 
expected to be approximated by expectation values for thermal state 
of the truncated Hamiltonian $H_0$ (below) if the boundary region 
$B$ is sufficiently large.
}
\label{fig:BoundaryRegion}
\end{figure}

In Ref.~\cite{Kliesch:2014}, it is proven that the correlations responsible for the distinguishability between the truncated and non-truncated thermal states are quantified by a \emph{generalized covariance}. For any two operators $O$ and $O'$, full-rank quantum state $\rho$, and parameter $\tau \in [0,1]$, the generalized covariance is defined as
\begin{equation}\label{eq:def_cov}
 \cov_\rho^\tau(O,O') 
 := \Tr\left(\rho^\tau O\, \rho^{1-\tau} O'\right) - \Tr(\rho\, O) \Tr(\rho \, O') \, ,
\end{equation}
and the average distinguishability of the two states measurements of some observable $O$ can provide reads as
\begin{equation}\label{eq:truncation_error_in_terms_of_cov}
\Tr\bigl[O \omega(H_0)\bigr] - \Tr\bigl[O\,\omega(H)\bigr]
= 2\int_0^1 \dd s\int_0^{\beta/2} \dd t \, \cov_{\omega_s}^{t/\beta} (H_{\boundary}, O),  
\end{equation}
where  $H_{\boundary}$ is the corresponding Hamiltonian term that couples $B$ and $C$, $H_0=H-H_I$ is the truncated Hamiltonian (see Fig.~\ref{fig:BoundaryRegion}) and $\omega_s=\omega(H(s))$ is the thermal state of the \emph{interpolating Hamiltonian} $H(s) := H - (1-s)\, H_{\boundary}$. Hence, \emph{the generalized covariance is the quantity that measures the response in a local operator of perturbing a thermal state and ultimately at what length scales temperature can be defined}.

Temperature is known to be a local quantity in a high temperature regime. More specifically, in Ref.~\cite{Kliesch:2014}, it is shown that for any local Hamiltonian there is a threshold temperature (that only depends on the connectivity of the underlying graph) above which the generalized covariance decays exponentially. Nevertheless, it is far from clear what occurs below the threshold, and, especially, at low temperatures ($\beta \gg 1$). Note that, in that case, the right hand side of the truncation formula (\ref{eq:truncation_error_in_terms_of_cov}) could be significantly different from zero since the integration runs up until $\beta/2$, while the covariance is expected to decay only algebraically for critical systems.

In this work we show that, for one dimensional translation-invariant systems, temperature is local for any $\beta$. Away from criticality, we rigorously bound the truncation formula (\ref{eq:truncation_error_in_terms_of_cov}) by mapping the generalized covariance to the contraction of a tensor network and exploiting some standard results in condensed matter. At criticality, we use some results from \emph{conformal field theory} \cite{DiFrancesco:1997,Cardy:2008}. Finally, the results in \cite{lluis}, where the equivalence of microcanonical and canonical ensembles is proven for translation-invariant lattices with short range interactions, render our results valid also when, instead of being canonical, (\ref{eq:thermal-state}), the global state $\omega(H)$ is, e.g., microcanonical. The latter is defined as an equiprobable mixture of all the energy eigenstates in a narrow energy window (see \cite{lluis} for details).

In condensed matter physics, this problem has been considered in the context of approximating the expectation values of infinite systems by finite ones, receiving the name of \emph{finite size scaling} (see, e.g., \cite{fss} and references therein). Nevertheless, the finite-size-scaling methods are more focused to find the values for the critical exponents and the transition temperature by observing how measured quantities vary for different lattice sizes.

The paper is organized as follows. In Sec.~2, the generalized covariance of 1D system is mapped to the contraction of a 2D tensor network. In Sec.~3, we show that temperature is local at \emph{non-zero temperature} ($\beta < \infty$) by identifying in the tensor network a gapped transfer matrix which leads to a clustering of correlations and ultimately to a clustering of the generalized covariance. In Sec.~4, locality of temperature is proven at \emph{zero temperature} ($\beta \to \infty$) using different methods for gapped and gapless systems. While transfer matrix arguments work satisfactory for gapped systems, conformal field theory results have to be used at criticality. In Sec.~5, all our results are illustrated in detail for the Ising model, for which we study in addition the behaviour of the generalized covariance and compute explicitly the physical distinguishability between the full and truncated Hamiltonians both at and off criticality. Finally, we conclude.

\section{Tensor network representation of the generalized covariance}
\subsection{Mapping the partition function of a $D$-dimensional quantum model to the contraction
of tensor network of $D+1$ dimensions}

Let us consider a system of spins described by a short range Hamiltonian.
The structure of the Hamiltonian is given by a graph $G(V,E)$. The spins correspond to the set of vertices $V$ and the
two-body interactions to the edges $E$. Such a Hamiltonian can be written as
\begin{equation}\label{eq:def_localH}
H=\sum_{u\in E} h_u
\end{equation}
where $h_u$ are the Hamiltonian terms acting non-trivially on the adjacent vertices of $u$.

\begin{figure}
\centering
\includegraphics[scale=0.7]{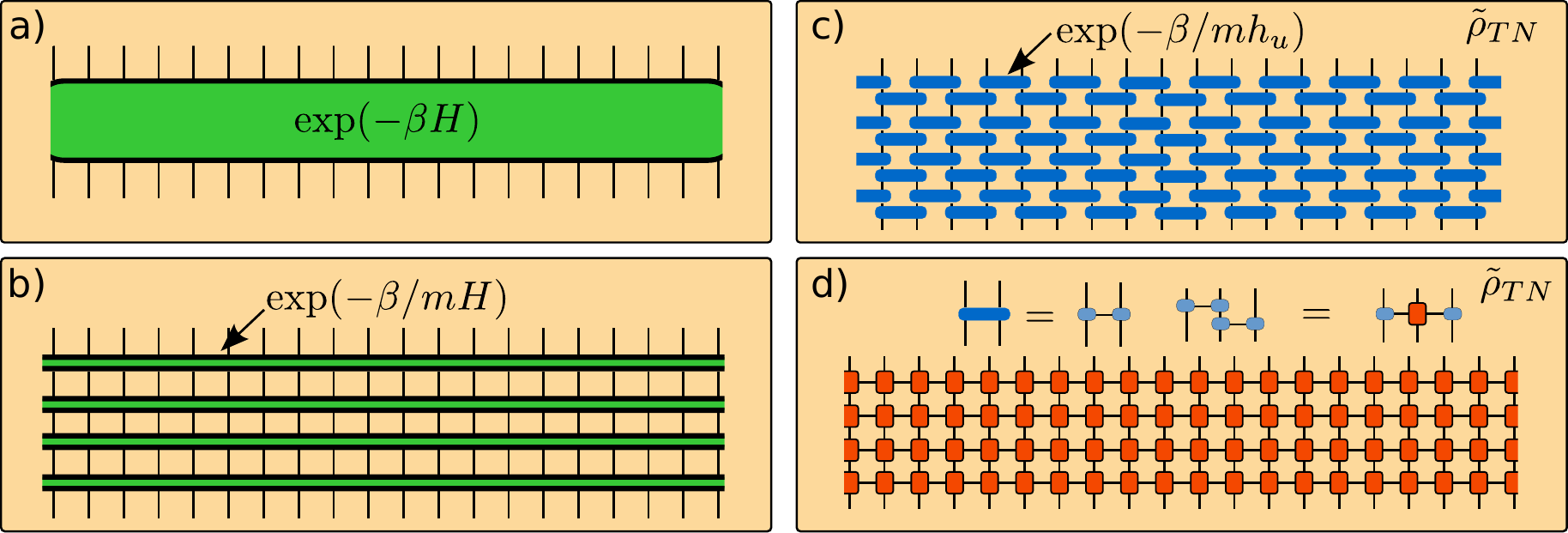}
\caption{Diagrammatic representation of
(a) the operator $\e^{-\beta H}$ with $H$ the Hamiltonian of a spin chain, 
(b) its decomposition $\left(\exp(-\beta/m H)\right)^m$,
(c) the tensor network that approximates $\e^{-\beta H}$ after performing the Trotter-Suzuki decomposition
 for a one dimensional short-ranged Hamiltonian
and (d) the same tensor network after a convenient arrangement of the tensors.
We use the Penrose notation: tensors are represented as geometric shapes, open legs represent their indices and legs connecting different tensors encode their contraction over the corresponding indices.
}
\label{fig:Trotter}
\end{figure}

In Ref.~\cite{hastings,Molnar:2014}, it is shown that, for any error $\varepsilon >0$, 
the matrix $\e^{-\beta H}$ of a local Hamiltonian can be approximated in one norm by
its Trotter-Suzuki expansion, 
\begin{equation}\label{eq:Trotter-expansion}
\tilde \rho_{TN}=\left(\left(\prod_{u\in E} \e^{-\frac{\beta}{2m} h_u }\right)\left(\prod_{v\in E} \e^{-\frac{\beta}{2m} h_v }\right)^\dagger\right)^m\, ,
\end{equation}
such that
\begin{equation}\label{eq:Trotter-approx}
\norm{\e^{-\beta H}-\tilde \rho_{TN}}_1 \le \varepsilon \norm{\e^{-\beta H}}_1\, ,
\end{equation}
where $m>360\beta^2|E|^2/\varepsilon$ and the products over $u$ and $v$ in Eq.~(\ref{eq:Trotter-expansion}) are realized in the same order.

To illustrate the previous approximation, let us consider in detail the one dimensional case: a spin chain with nearest neighbour interactions.
By decomposing in the standard way the Hamiltonian in its odd and even terms,
the tensor network $\tilde \rho_{TN}$ becomes in this case
\begin{equation}
\tilde \rho_{TN}=\left(\e^{-\frac{\beta}{2m} H_\textrm{odd} }\e^{-\frac{\beta}{m} H_\textrm{even} }
\e^{-\frac{\beta}{2m} H_\textrm{odd} } \right)^m\, ,
\end{equation}
where $H_\textrm{odd/even} = \sum_{u\in \textrm{odd/even}} h_u$ and $H=H_\textrm{odd}+H_\textrm{even}$.

Let us think about each $\exp(\beta/m h_u)$ as a tensor.
In this way, $\tilde \rho_{TN}$ can be seen as the contraction
of several of such tensors, that is, a \emph{tensor network}
see Fig.~\ref{fig:Trotter} (c).
Starting from a one dimensional quantum system,
$\tilde \rho_{TN}$ can be interpreted as a \emph{tensor network} spanning two dimensions, with the extra dimension of length $m$.
We will refer to this extra dimension as 
the $\beta$ direction, while the original dimension will be called spatial direction.

In Fig.~\ref{fig:Trotter} (a), the diagrammatic representation of $\tilde \rho_{TN}$ is presented.
Its tensors can be decomposed and arranged in order to form a square lattice of elementary tensors as 
shown in Fig.~\ref{fig:Trotter} (b).

\subsection{Generalized covariance as the contraction of tensor networks}
The expectation value of a local operator is given by
\begin{equation}
\langle O \rangle=\Tr(O \omega(H))=\frac{\Tr(O \e^{-\beta H})}{\Tr(\e^{-\beta H})}\, .
\end{equation}
By using Eq.~(\ref{eq:Trotter-approx}), the fact that $\norm{\e^{-\beta H}}_1=Z$ and some elementary algebra, the expectation
value of a local operator $A$ can be approximated by the ratio between the contraction of 
two tensor networks
\begin{equation}
\left|\langle O \rangle -\frac{\Tr(O \tilde \rho_{TN})}{\Tr( \tilde \rho_{TN})} \right|\le 2 \norm{O}_\infty \varepsilon\, .
\end{equation}
This is represented diagramatically in Fig.~\ref{fig:expectation-values}(a).

\begin{figure}
\centering
\includegraphics[scale=0.8]{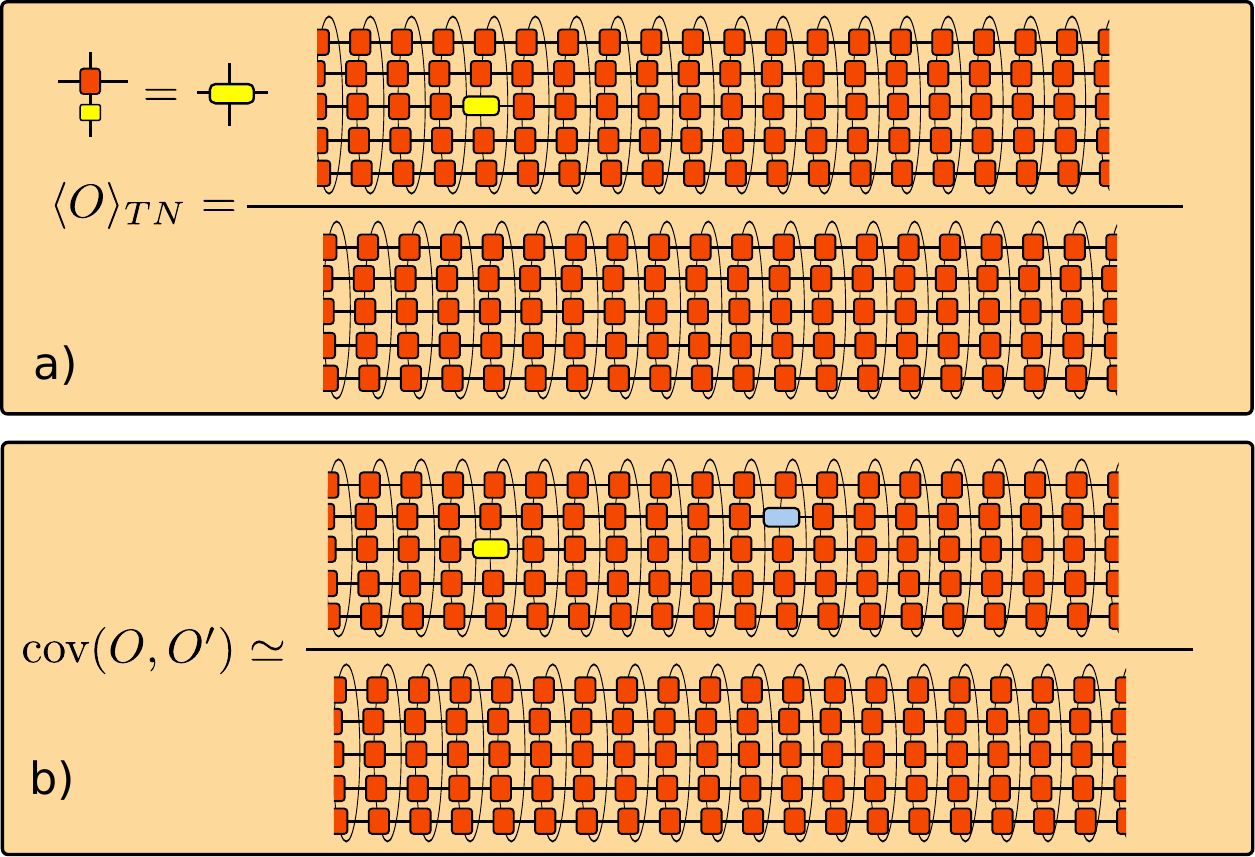}
\caption{Diagramatic representation of (a) the expectation value of a one site operator 
and (b) the generalized covariance (a two-point correlation function) between two one site operators.
In both cases, the final result is computed as the ratio between the contraction of two tensor networks.}
\label{fig:expectation-values}
\end{figure}

The generalized covariance can be rewritten as
\begin{equation}
\cov_{\omega_s(\beta)}^{\tau} (O,O') = 
\frac{\Tr \left( \tilde O \e^{-\tau \beta H} \tilde O' \e^{-(1-\tau)\beta H} \right)}{\Tr \left( \e^{-\beta H}\right)}\, ,
\end{equation}
where $\tilde O=O- \Tr(O\omega(H))$ for any operator $O$.
Hence, in a similar way as it has been made for the expectation values, the generalized covariance 
can be also approximated as the ratio between two tensor network
contractions
as shown in Fig.~\ref{fig:expectation-values}(b)
\begin{equation}
\left|\cov_{\omega_s(\beta)}^{\tau} (O,O') 
-\frac{\Tr(\tilde O \tilde \rho_{TN}^\tau \tilde O' \rho_{TN}^{(1-\tau)})}{\Tr( \tilde \rho_{TN})}
\right|\le 2 \norm{O}_\infty \norm{O'}_\infty \varepsilon \, .
\end{equation}
From this perspective, the generalized covariance can be seen as a two point correlation function on
a 2 dimensional lattice in which $\tau m$ is the separation in the $\beta$ direction
and the distance between the non-trivial supports of $O$ and $O'$ is the separation in the spatial direction
(see Fig.~\ref{fig:expectation-values}).

This construction can be generalized to approximate expectation values of local operators
and $n$-point correlation functions of a $D$ dimensional quantum model by
the ratio of the contraction of two $D+1$ dimensional tensor networks. 

\subsection{Transfer matrices}
It is also very useful to define two extra objects: the \emph{transfer matrices} along the spatial $T$ and $\beta$ directions $T_\beta$.
The first is obtained by contracting a column of the elementary tensors of the network, while the second is obtained by contracting several rows of elementary tensors, see Fig.~\ref{fig:spatial-transfer-matrix}.

The number of rows that need to be contracted in order to obtain the transfer matrix in the $\beta$ direction, $T_\beta$, is chosen 
such that its spectral gap between the largest and second-largest eigenvalues is independent of both $\beta$ and $m$. 
This can be achieved by contracting 
 $m/\beta$ rows, leading to a transfer matrix with two largest eigenvalues $\mu_1$ and $\mu_2$
\begin{equation}
\frac{\mu_2}{\mu_1}=\e^{-\Delta}
\end{equation}
where $\Delta$ is the gap of the Hamiltonian.

\begin{figure}
\centering
\includegraphics[scale=0.9]{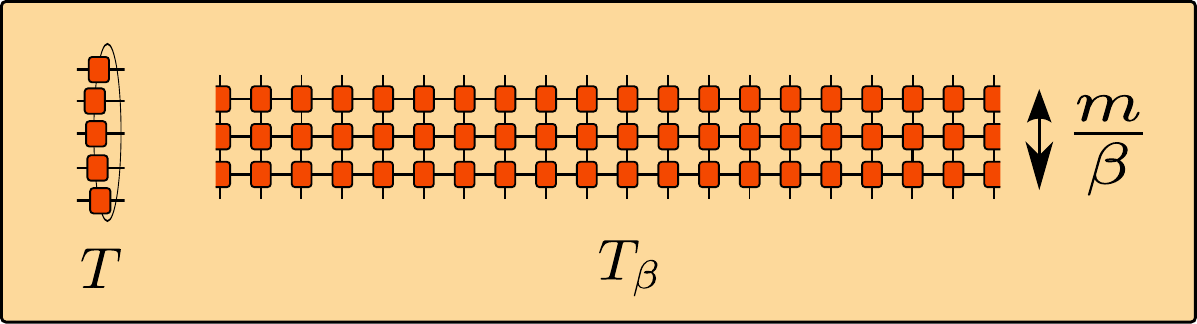}
\caption{Diagramatic representation of the transfer matrix in the spatial direction (left) and the $\beta$ direction (right).}
\label{fig:spatial-transfer-matrix}
\end{figure}

\section{Locality of temperature at non-zero temperature}\label{section:finite-temperature}

Let us consider now the case in which $\beta$ is of order one.
The physical distinguishability in $A$ between the full and the truncated Hamiltonians
can be bounded by 
\begin{equation}\label{eq:trivial-bound}
\hspace{-2.5cm} \Tr\bigl[O \omega(H_0)\bigr] 
	 - \Tr\bigl[O\,\omega(H)\bigr]
    \hspace{-0.1cm}\leq\hspace{-0.1cm} \beta\hspace{-0.1cm} \int\limits_{0}^1 \hspace{-0.1cm}\dd s \max_{\tau\in[0,1]} 
    \cov_{\omega_s}^{\tau} (H_{\boundary}, O)
    \hspace{-0.1cm}\leq\hspace{-0.1cm} 2\beta \hspace{-0.1cm}\int\limits_{0}^1 \hspace{-0.1cm}\dd s  \max_{i} \max_{\tau\in[0,1]}\cov_{\omega_s}^{\tau} (h_i, O),
\end{equation}
where  $H_I=h_L+h_R$, $i\in\{L,R\}$, and we have used the linearity of the generalized covariance with respect of its operators.

Without loss of generality let us assume that the term of $H_I$ that maximizes the generalized covariance is the one of the left, $h_L$. Hence, the quantity to bound is $\cov_{\omega_s}^{\tau} (h_L, O)$.
In order to do so, let us rewrite it as
\begin{equation}\label{eq:3pointfunction}
\hspace{-2.3cm}\cov_{\omega_s}^{\tau} (h_L, O)\hspace{-0.1cm}=\hspace{-0.1cm}
\frac{\bra{1_L} X_s T^{\ell_B} Y T^{\ell_B} T_s \ket{1_R}}{Z_s} 
-\frac{\bra{1_L} T_s T^{\ell_B} Y T^{\ell_B} T_s\ket{1_R}}{Z_s}
\frac{\bra{1_L} X_s T^{2\ell_B+1} T_s \ket{1_R}}{Z_s}~
\end{equation}
where $Z_s:=\bra{1_L}T_s T^{2\ell_B+1} T_s\ket{1_R}$ is the partition function, 
$T$ is the transfer matrix in the spatial direction,
see Fig.~\ref{fig:spatial-transfer-matrix} (left),
and $\ket{1_{L/R}}$ is the left/right dominant eigenvector of $T$ i.~e.\ the eigenvector associated
to its largest eigenvalue. 
The matrix $T_s$ is the transfer matrix corresponding to the boundaries $BC$
where the elementary tensors of the network are different from the rest for $s< 1$ and $T_{1}=T$.
The matrix $Y$ corresponds to the slice of the region $A$
where the operator $O$ is supported, 
and the matrix $X_s$ is the transfer matrix $T_s$ with the
insertion of the operators $h_L$ located at a distance $\tau \beta$ from $O$ in the transverse direction.
The diagrammatic representations of the matrices $X_s$, $Y$ and $T_s$ are shown in Fig.~\ref{fig:gapped-systems}.
To simplify the calculations, the transfer matrix can always be normalized such that its dominant
eigenvalue is $\lambda_1=1$.

\begin{figure}
\centering
\includegraphics[scale=.9]{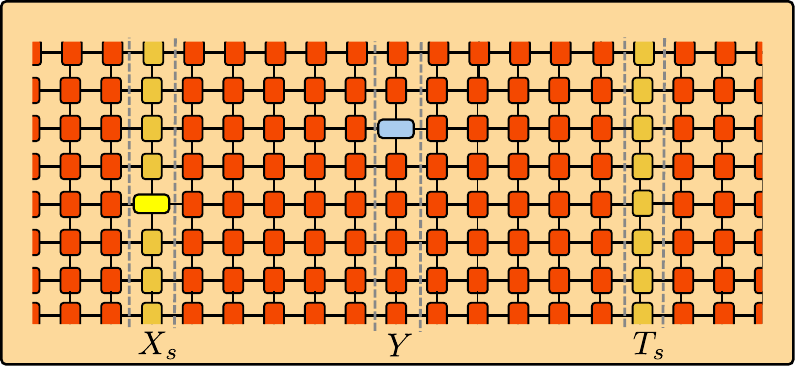}
\caption{Diagrammatic representation of a 3-point correlation function of the matrices $X_s$, $Y$ and $T_s$ for a system at zero temperature. The network is infinite in both directions.}
\label{fig:gapped-systems}
\end{figure}

To bound the generalized covariance (\ref{eq:3pointfunction}) it is useful to 
rewrite it in terms of 2-point correlation functions of the uniform system ($s=1$)
\begin{eqnarray}\label{eq:GeCov-as-2PCF}
\cov_{\omega_s}^{\tau} (h_L, O)&=&
\frac{\cov_T(\ell_B;X_s,YT^{\ell_B}T_s)}{Z_s}\nonumber \\
&-&\frac{\av{X_s}_T\av{T_s}_T}{Z_s}\frac{\cov_T(\ell_B;T_s,YT^{\ell_B}T_s)}{Z_s} \nonumber \\
&+&\frac{\av{X_s}_T\av{YT^{\ell_B} T_s}_T}{Z_s}\frac{\cov_T(2\ell_B+1;T_s,T_s)}{Z_s} \\
&-&\frac{\av{T_s}_T\av{YT^{\ell_B} T_s}_T}{Z_s}\frac{\cov_T(2\ell_B+1;X_s,T_s)}{Z_s} \nonumber \\
&-&\frac{\cov_T(\ell_B;T_s,YT^{\ell_B}T_s)}{Z_s}\frac{\cov_T(2\ell_B+1;X_s,T_s)}{Z_s}\, ,\nonumber 
\end{eqnarray}
where
\begin{equation}\label{eq:expectation-value-uniform}
\av{X}_T:=\bra{1_L} X \ket{1_R},
\end{equation}
\begin{equation}\label{eq:2pointfunction-uniform}
\cov_T (\ell; X, Y):=\bra{1_L} X T^{\ell}Y \ket{1_R}-
\bra{1_L} X \ket{1_R}\bra{1_L} Y \ket{1_R}\, ,
\end{equation}
and we have used that $\bra{1_L} T^{\ell_B+1} T_s\ket{1_R}=\av{T_s}_T$ and $\cov_T(\ell_B;X_s,T^{\ell_B+1}T_s)=\cov_T(2\ell_B+1;X_s,T_s)$.

In short range one dimensional systems, \emph{the absence of phase transitions at non-zero temperature} \cite{Gelfert2001}
\emph{implies that the transfer matrix $T$ is gapped, with a gap related to 
the spatial correlation length} as
\begin{equation}\label{eq:corr_length}
\xi =-\left( \ln|\lambda_2| \right)^{-1}>0\, ,
\end{equation}
where $\lambda_{2}$ is the second largest eigenvalue of the transfer matrix $T$.

For gapped transfer matrices, the 2-point correlation function (\ref{eq:2pointfunction-uniform}) can be proven to be upper-bounded by
\begin{equation}
|\cov_T(\ell;X,Y)|\le \norm{X \ket{1_L}} \norm{Y\ket{1_R}} \e^{-\ell/\xi}\, .
\end{equation}
The complete proof of the previous statement can be found in
Lemma~\ref{lemma:covarianceIS} of the \ref{appa2}.

Furthermore, Lemma~\ref{lemma:covarianceIS} allows us to bound all the terms in Eq.~(\ref{eq:GeCov-as-2PCF}), and, as it is shown in \ref{appfin}, the following inequality holds for the left hand side of Eq.~(\ref{eq:trivial-bound}):
\begin{equation}
    \Tr \left[O\,\omega(H_{AB})\right] 
	 - \Tr\left[O\,\omega(H)\right]
    \le 2 \beta c \norminf{h_L}\norminf{O} e^{-\ell_B /\xi} +\mathcal{O}(\e^{-2\ell_B / \xi}), 
\end{equation}
where $c=1+\int_0^1 \dd s \sigma_L(s)$, with $\sigma_{L}(s)=(\bra{1_{L}}T_sT_s^\dagger\ket{1_{L}}-|\av{T_s}_T|^2)^{1/2}/|\av{T_s}_T|$ being the relative standard deviation of $T_s$ on the left dominant eigenstate $\ket{1_L}$. The quantity $c$ is a constant of order one that depends on the model considered. Hence, the temperature is proven to be intensive for any one dimensional translationally invariant model at non-zero $\beta$.

\section{Absolute zero temperature}
\subsection{Gapped systems} \label{pupulik}
Given a Hamiltonian with gap $\Delta$, here we study the regime in which $\beta^{-1} \ll \Delta$. 
This implies that the lattice in its $\beta$ direction is much larger than the correlation length
$\beta \gg \xi_\beta$, with
\begin{equation}
\xi_\beta:=\left(\ln\left(\frac{\mu_1}{\mu_2}\right)\right)^{-1}=\Delta^{-1} \, .
\end{equation}
In the limit of temperature tending to zero, the 2D network that represents the partition function becomes infinite in the
$\beta$ direction (see Fig.~\ref{fig:gapped-systems}).

In order to see that the temperature is also local in this case,
let us decompose the integral over $t$ of the generalized covariance into two pieces
\begin{eqnarray}
\hspace{-1.7cm}\Tr\bigl[O\,\omega(H_{AB})\bigr] 
	 &-& \Tr\bigl[O\,\omega(H)\bigr]=
2\int_0^1 \dd s\int_0^{\beta/2} \dd t \, \cov_{\omega_s}^{t/\beta} (H_{\boundary}, O)\nonumber \\
&=&
2\int_0^1 \dd s\int_0^L \dd t \, \cov_{\omega_s}^{t/\beta} (H_{\boundary}, O)+
2\int_0^1 \dd s\int_L^{\beta/2} \dd t \, \cov_{\omega_s}^{t/\beta} (H_{\boundary}, O)
\end{eqnarray}
where $L$ is a cut-off that will be chosen afterwards to minimize a bound on the right hand side,
and $\beta$ will be made to tend to infinity.

Concerning the integral over $0\le t\le L$, we will exploit the fact that
the system is gapped, and hence \emph{its ground state is known to have a finite correlation length} $\xi$ in the spatial direction and to be \emph{represented by a Matrix Product State} of bond dimenson $D$,
with $D\propto \textrm{poly}(\xi)$ \cite{Perez-Garcia2007,Pirvu2012a,Verstraete2006}.
As argued in the previous section, a finite correlation length guarantees a gap in the transfer
matrix in the corresponding direction.
By performing an analogous calculation to the one described in the previous section, one obtains
\begin{equation}
\int_0^1 \dd s\int_0^L \dd t \, \cov_{\omega_s}^{t/\beta} (H_{\boundary}, O)\le 
2c\norminf{h}L \norminf{O}e^{-\ell_B /\xi} \, .
\end{equation}

The second integral over $t>L$ can be bounded by
taking the transfer matrix in the $\beta$ direction
which is also gapped for gapped Hamiltonians.
More specifically, the generalized covariance can be written as
\begin{equation}
\cov_{\omega_s}^{t/\beta} (H_{\boundary}, O)=\frac{\bra{GS} O T_\beta^{t} H_\boundary \ket{GS}}{\bra{GS} T_\beta^{t} \ket{GS}}-
\frac{\bra{GS} O T_\beta^{t} \ket{GS}}{\bra{GS} T_\beta^{t} \ket{GS}}\frac{\bra{GS} T_\beta^{t} H_\boundary \ket{GS}}{\bra{GS} T_\beta^{t} \ket{GS}}\, .
\end{equation}
where we have identified $T_\beta$ as the transfer matrix for which the ground state
of the Hamiltonian $\ket{GS}$ is its dominant eigenvector.
As previously, we make use of Lemma~\ref{lemma:covarianceIS} in the \ref{appa2} and obtain
\begin{equation}
|\cov_{\omega_s}^{t/\beta} (H_{\boundary}, O)|\le \norminf{H_{\boundary}} \norminf{O}\e^{-t/\xi_\beta} .
\end{equation}
The integration is then bounded by
\begin{equation}
\lim_{\beta \to \infty}\int_L^{\beta/2} \dd t \int_0^1 \dd s\, \cov_{\omega_s}^{t/\beta} (H_{\boundary}, O)\le 
 \xi_\beta \norminf{H_{\boundary}} \norminf{O}e^{-L /\xi_\beta} \, .
\end{equation}

Putting the previous bounds together, and after an optimization over $L$, we get
\begin{equation}
\hspace{-1cm}\Tr \left[O\,\omega(H_{AB})\right] 
	 - \Tr\left[O\,\omega(H)\right]\le 
4\norminf{O}\norminf{H_{\boundary}}\,c \
\xi_\beta \left(1+ \frac{\ell_B}{\xi}-\ln c\right)\e^{-\ell_B/ \xi }\, ,
\end{equation}
showing that temperature can be locally assigned to subsystems for arbitrarily large $\beta$ and gapped Hamiltonians.

\subsection{Criticality}
A system at zero temperature is said to be critical when the gap between the energy ground state (space) and
the first excited state closes to zero in the thermodynamic limit.
The critical exponents $z\nu$ control how the spectral gap $\Delta$ tends to zero
\begin{equation}
\Delta \propto N^{-z\nu}\, ,
\end{equation}
where $N$ is the system's size and $\nu$ is the critical exponent that controls the divergence of the correlation length
\begin{equation}
\xi \propto N^{\nu}\, .
\end{equation}

The previous divergences are a signature of the \emph{scale invariance} that the system experiences at criticality. 
If the critical exponent $z=1$, there is a further symmetry enhancement
and the system becomes \emph{conformal invariant}.
The group of conformal transformations includes, in addition to scale transformations, translations and
rotations.

In 1+1 dimensions, conformal symmetry completely dictates how correlation functions behave and how local expectation values of local observables of infinite systems differ from those taken for finite ones.
Hence, conformal field theory establishes that
\begin{equation}
\Tr\bigl[O\,\omega(H_{AB})\bigr] 
	 - \Tr\bigl[O\,\omega(H)\bigr]
    \simeq \frac{1}{\ell_B^y}
\end{equation}
up to higher order terms, where $y$ is the \emph{scaling dimension} of the operator $H_\boundary$
\cite{Cardy1984, Cardy1986}. 
If $H_\boundary$ is a standard Hamiltonian term, in the sense that the system is homogeneous, its leading scaling dimension is $y=2$.

Once more, we see that by increasing the buffer region temperature can be arbitrarily well assigned.
\newpage
\section{A case study: The Ising chain}

Now we illustrate our results for the quantum Ising chain, which is described by  the Hamiltonian
\begin{equation}\label{eq:numerical-example}
H_N = \frac{1}{2} \sum_{i=1}^{N-1} \sigma_x^i \otimes \sigma_x^{i+1}-\frac{h}{2} \sum_{i=1}^N \sigma_z^i,
\end{equation}
where $\sigma_x^i$ and $\sigma_z^i$ correspond to the Pauli matrices, $h$ characterizes the strength of the magnetic field and $N$ is the number of spins. Notice that the interactions in the above Hamiltonian are of finite range, a crucial assumption in our derivations, see (\ref{eq:def_localH}). 
This model has a quantum phase transition at $h=1$, so it well exemplifies the different regimes discussed above: criticality (only  at zero temperature) and away from it (for zero and non-zero temperatures).

\subsection{Generalized Covariance}
First of all, as in the previous sections, we split the chain in three regions, which are shown in Fig.~\ref{gcov-setting2}. For such a splitting, and in the context of Eq. (\ref{eq:truncation_error_in_terms_of_cov}), we compute the generalized covariance $\cov_{\omega_s}^{t/\beta}(O,O')$ taking for $O$ a local operator in $A$, $O=\sigma_z^{N/2}$, and for $O'$ the boundary Hamiltonian between $B$ and $C$, $O'=H_\boundary$, given by
\begin{equation}
H_\boundary =\frac{1}{2}\left(\sigma_x^{N/2-3} \otimes \sigma_x^{N/2-2}+\sigma_x^{N/2+2} \otimes \sigma_x^{N/2+3}\right)\text{.}
\end{equation}
\begin{figure}[H]
\centering
\includegraphics[scale=0.5]{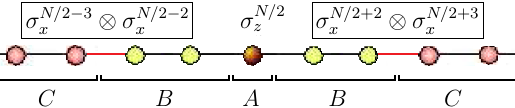}
\caption{Scheme of the subsystem $A$, the boundary region $B$ and their environment $C$. The local operator $\sigma_z^{N/2}$ acts on the subsystem $A$ and the interaction term $H_\boundary$ corresponds to the red lines (connection between the subsystems $AB$ and $C$).}\label{gcov-setting2}
\end{figure}
In order to compute $\cov_{\omega_s}^{t/\beta}(\sigma_z^{N/2},H_\boundary)$, we first diagonalize the Hamiltonian (\ref{eq:numerical-example}) using standard techniques from statistical mechanics, such as the Jordan-Wigner and the Bogoliubov transformation (see \ref{appa}). Once the Hamiltonian is diagonalized, we can straightforwardly construct the corresponding thermal state for every large but finite $N$, and compute $\cov_{\omega_s}^{t/\beta}(\sigma_z^{N/2},H_\boundary)$  using expression (\ref{eq:def_cov}).

Figure \ref{Fig1} shows $\cov_{\omega_s}^{t/\beta}(\sigma_z^{N/2},H_\boundary)$ as a function of $t/\beta$ for several temperatures ($\beta=5,20,1000$), and for $h=0.9,1$ (i.e., near and at criticality). We take $N=40$, which already describes well the thermodynamic limit (recall that we are only interested on the local state, and that the correlations decay exponentially). The area below the curves correspond to the first integral in (\ref{eq:truncation_error_in_terms_of_cov}), which measures how well the local state in $A$ can be approximated by a thermal state in $AB$. 

The results in Fig.~\ref{Fig1} are in agreement with properties (i) and (ii) from Lemma \ref{lemma:covariancePBC} in \ref{appa2}. The first property implies that the covariance is symmetric with respect to $t=\beta/2$, and it follows by taking $l=t$ and $n=\beta$ in (\ref{eq:symmetric}). Second, property (ii) implies that it is bounded by a convex function of $t$ with a maximum at $t=0$ and $t=\beta$ and with a minimum at $t=\beta/2$. Therefore, the covariance satisfies the bound (\ref{eq:convex-bound}).


On the other hand, the covariance is not monotonic in $s$ (see Fig.~\ref{Fig2}). This is somehow counterintuitive, as it shows that the outcomes of two observables with no overlapping support (located in $A$ and in the intersection between $B$ and $C$) do not always become more correlated as $s$, which quantifies the strength of the interaction between $B$ and $C$, increases. 

\begin{figure}[H]
\small
\hspace{1.5cm}\includegraphics[scale=0.7]{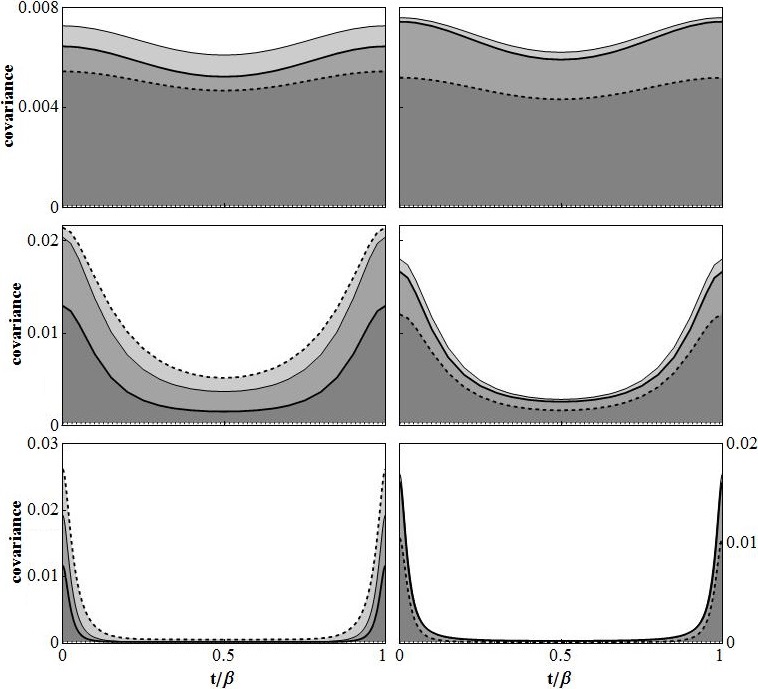}

\caption{Generalized covariance as a function of $t$ for different values of $s$: $s=0,1/3,2/3,1$ for the dotted, dashed, black and thick lines. The figures correspond to inverse temperature $\beta=5$ (top), $\beta=20$ (at the middle) and $\beta=100$ (bottom) and field strength $h=0.9$ (left) and $h=1$ (right). The grey area below the curves corresponds to the first integral of Eq. (\ref{eq:truncation_error_in_terms_of_cov}).}\label{Fig1}
\end{figure}

\begin{figure}[H]
\small
\hspace{1.5cm}\includegraphics[scale=0.7]{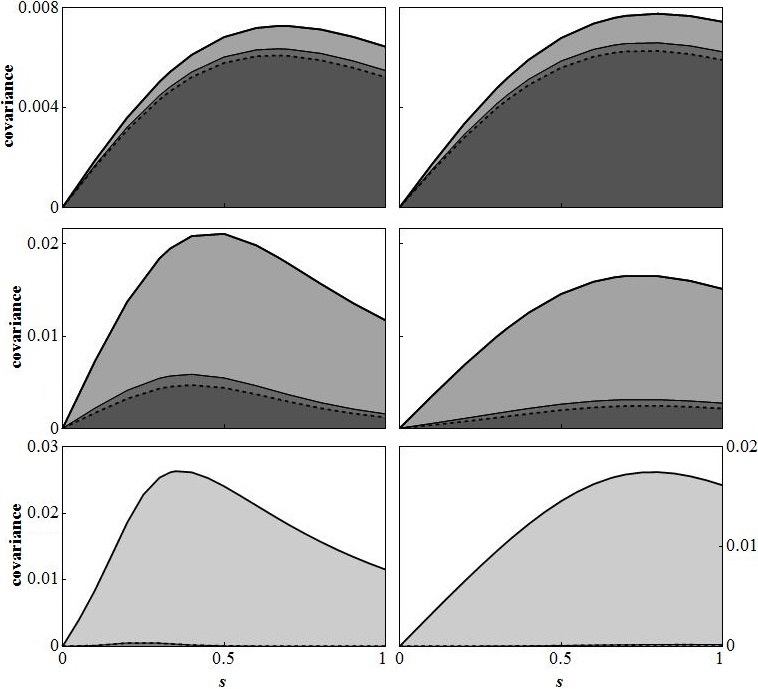}

\caption{Generalized covariance as a function of $s$ for different values of $t$: $t/\beta=0,1/3,1/2$ for the thick, black and dashed lines. The figures correspond to inverse temperature $\beta=5$ (top), $\beta=20$ (at the middle) and $\beta=100$ (bottom) and field strength $h=0.9$ (left) and $h=1$ (right). The grey area below the curves corresponds to the second integral of Eq. (\ref{eq:truncation_error_in_terms_of_cov}). Notice that, due to the symmetry in $t$, the values $t/\beta=2/3,1$ are also considered.}\label{Fig2}
\end{figure}

\subsection{Locality of temperature in the quantum Ising chain}\label{localityIsing}

In our analytical findings, the generalized covariance naturally appeared as a tool to solve the locality of temperature problem, see (\ref{eq:truncation_error_in_terms_of_cov}). This motivated the previous section, where we studied its properties in the context of the quantum Ising chain. Nevertheless, in order to obtain (\ref{eq:truncation_error_in_terms_of_cov}), one still needs to integrate $\cov_{\omega_s}^{t/\beta}(O,H_{\boundary})$ over $s$ and $\tau$. While this approach is useful when dealing with arbitrary generic systems, here we are dealing with a specific model that is furthermore solvable, so we can take a more direct approach. Concretely, we first compute 
\begin{equation}\label{eq:rhoA}
\rho_A={\rm Tr_{BC}}(\omega(H)), \hspace*{10mm} {\rm with} \hspace*{3mm} H\equiv H_{\infty}
\end{equation} 
and
\begin{equation} \label{eq:rhoA'}
\rho'_A={\rm Tr_{BC}}(\omega(H_{AB})), \hspace*{10mm} {\rm with} \hspace*{3mm} H_{AB}\equiv H_{N},
\end{equation}
for different sizes $N$ of the region $AB$. Secondly, we measure the distinguishability between such states via the quantum fidelity, which is advantageous for  computational reasons. For two states, $\rho_A'$ and $\rho_A$, the fidelity is defined as \cite{mikeike}
\begin{equation}
F[\rho_A,\rho'_{A}]=\tr\left[\sqrt{\sqrt{\rho_A}\rho'_{A}\sqrt{\rho_A}}\right].
\end{equation}
It satisfies $0\leq F\leq1$ and $F[\rho_A,\rho'_{A}]=1$ if and only if $\rho_A=\rho'_{A}$. In order to relate this approach to our previous considerations, we note the following relation between the trace distance, $D[\rho_A,\rho'_{A}]$, and $F(\rho_A,\rho'_{A})$, given in \cite{FvdG},
\begin{equation}
1-F\leq D\leq \sqrt{1-F^2}.
\end{equation}
Therefore, the fidelity provides us with upper and lower bounds to (\ref{dist2}). In particular, when $D[\rho_A,\rho'_{A}]\rightarrow 0$ then $F[\rho_A,\rho'_{A}]\rightarrow 1$, and in that case we say that the temperature is locally well defined.

From now on, we take for $A$ a two spin subsystem, an infinite chain as the total system, and we compute $F(\rho_A,\rho'_{A})$ as a function of the size of $AB$, with $N=2+2l_B$, and the different parameters of the Hamiltonian. 


In order to compute $\rho_A$ and $\rho_A'$, it is convenient to apply the Jordan-Wigner transformation to (\ref{eq:numerical-example}), which maps spin operators $\sigma^i_{x,y,z}$ to fermionic operators $a_i, a_i^{\dagger}$ (see \ref{appa} for details). The Hamiltonian (\ref{eq:numerical-example}) takes then the form,
\begin{equation}\label{eq:numerical-example02}
H_N=\sum_{i,j=1}^N A_{ij} a_i a_j^\dagger+\frac{1}{2}\sum_{i,j=1}^N B_{ij}(a_i^\dagger a_j^\dagger-a_i a_j),
\end{equation} 
which is quadratic in terms of the fermionic operators. It follows that thermal states, as well as their local states, are gaussian operators. Therefore it is possible to describe them by their covariant matrix, whose size is only $\mathcal{O}(N^2)$.  This allows us to compute   $\rho'_{A}$ in (\ref{eq:rhoA'}) for finite but large $l_B$; while in the limit $N\rightarrow \infty$, i.e. to compute $\rho_A$ in (\ref{eq:rhoA}),  we rely on the analytical results from \cite{Barouch}. The explicit calculations are done in \ref{appa}. 

\subsubsection{Non-zero temperatures}
\hfill \break
Figure \ref{Fig3} shows $F(\rho_A,\rho'_{A})$ as a function of $\beta$ and $h$, for $N=4$ (left) and $N=20$ (right). Recall that $N$, with $N=2+2l_B$, defines the size of the boundary region which is used to approximate $\rho_A$ by $\rho'_{A}$. Even if the boundary is small, $N=4$, the fidelity is close to 1 for all values of $\beta$ and $h$, and thus the temperature is locally well-defined. As expected, $F(\rho_A,\rho'_{A})$ increases with $N$ (see Fig.~\ref{Fig5}).

We also observe in Fig.~\ref{Fig3} that the fidelity becomes minimal near $h=1$, which is the phase transition point. As $N$ increases, this minimum is shifted to $h=1$. At this point the spatial correlations also increase, which suggests a relation between both quantities. 

In order to further explore this connection, we compute the scaling of $F(\rho_A,\rho'_{A})$ with $N$, and compare it to the decay of the correlations. The behaviour of $F(\rho_A,\rho'_{A})$ is plotted in Fig.~\ref{Fig5}, which clearly shows that the fidelity follows an exponential law with $N$, given by
\begin{equation}
F(\rho_A,\rho'_{A})\sim 1-e^{-\frac{N}{2\xi_S}}\text{,}
\end{equation}
where $\xi_S$ is a parameter that characterizes the slope of the function. On the other hand,  the correlations between a local observable in $A$, $\sigma_z^i$, and one in the intersection of $B$ and $C$, $\sigma_z^{i+d}$,
\begin{equation}\nonumber
\corr(\sigma_z^i, \sigma_z ^{i+d}) =\langle \sigma_z^i \sigma_z ^{i+d} \rangle - \langle \sigma_z^i \rangle \langle \sigma_z^{i+d} \rangle \text{,}
\end{equation}
can be obtained through the two-spin correlation function $\langle \sigma_z^i \sigma_z ^{i+d} \rangle$ in  \cite{Barouch}. Their asymptotic behaviour is also exponential with $d$, 
\begin{equation}\nonumber
\corr(\sigma_z^i, \sigma_z ^{i+d}) \sim e^{-\frac{d}{\xi}}\text{,}
\end{equation}
where $\xi$ is the correlation length. Now, identifying $d$, the distance between particles, with $N/2$, which is roughly the size of  $B$, we obtain from the numerical results in Fig.~\ref{Fig6} the following simple relation,
\begin{equation}\nonumber
\xi = 2\xi_S \text{.}
\end{equation}
Roughly speaking, the quality of the approximation $\rho_A'$ is directly related to the strength of the correlations in the system. This relation is in good agreement with previous considerations in \cite{Michael}, where the correlation length is related to the error of the cluster approximation \cite{Michael,MichaelII}.

In summary, temperature can be assigned to the local system for all  $h$ and non-zero $\beta$ by taking a small boundary region (with $N \geq 4$, and thus $l_B\geq 2$). We have shown that this is directly connected to the exponential decay of the correlations with the distance, which makes the local state of a thermal state only be sensible to its closest boundary. 

\begin{figure}[H]
\centering
\includegraphics[scale=0.5]{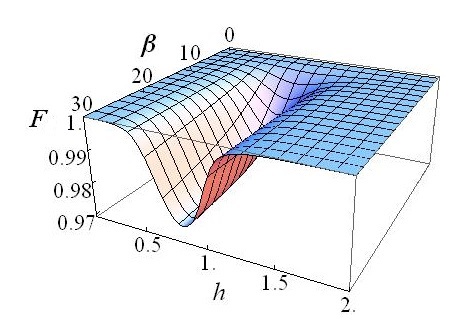}
\hspace{0.2cm}
\includegraphics[scale=0.5]{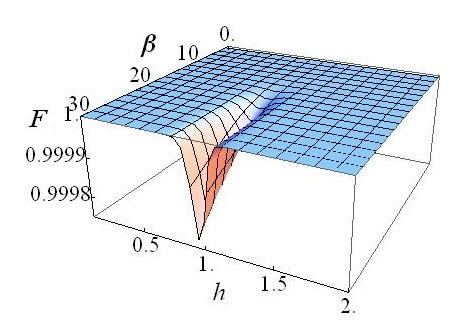}
\caption{Fidelity $F(\rho_A,\rho'_{A})$ as a function of $\beta$ and the strength of the magnetic field $h$ for $N=4$ (left) and $N=20$ (right). The temperature is locally well-defined provided that $F(\rho_A,\rho'_{A})\approx 1$.}\label{Fig3}
\end{figure}


\begin{figure}[H]
\centering
\includegraphics[scale=0.7]{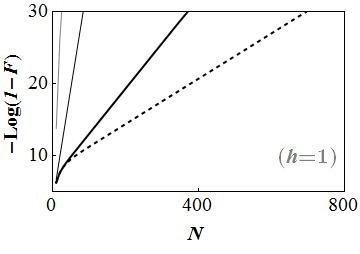}\hspace{0.5cm}
\includegraphics[scale=0.7]{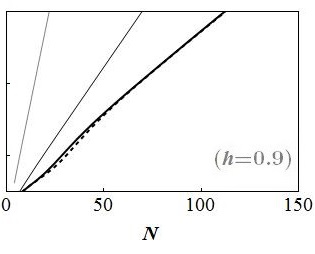}\\
\caption{Function of fidelity, $-\log(1-F[\rho_A,\rho'_{A}])$, as a function of $N$ for $h=1$ (left) and $h=0.9$ (right). The inverse temperature is $\beta=5,20,100,200$ for the grey, black, thick and dashed lines.}\label{Fig5}
\end{figure}

\begin{figure}[H]
\centering
\includegraphics[scale=0.5]{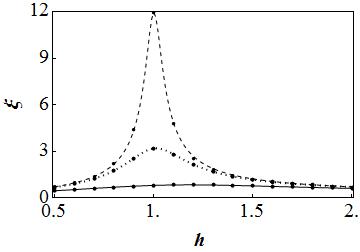}
\caption{Correlation length, $\xi$, as a function of $h$. The black spots correspond to the numerical values for $2\xi_S$. The inverse temperature is $\beta=5,20,75$ for the black, dotted and dashed lines.}\label{Fig6}
\end{figure}

\subsubsection{Absolute zero temperature}
\hfill \break
The same conclusions apply at zero temperature, as the fidelity is also close to 1 for all $h$ and $N\geq 4$. It also has a minimum near the critical point.

Nonetheless, the scaling of the fidelity (or more precisely $1-F$) with $N$ can differ from the previous case. While the scaling is generally exponential at zero temperature, it becomes a power law  at the phase transition point (see Fig.~\ref{Fig10}),
\begin{equation}
F(\rho_A,\rho'_{A})\sim 1-N^{-C_s} \text{.}
\end{equation}
This type of decay is also obtained for the correlations as  a function of the distance, which again shows a direct connection between the quality of the approximation (quantified by $F(\rho_A,\rho'_{A})$) and the strength of the correlations.


\begin{figure}[H]
\centering
\includegraphics[scale=0.7]{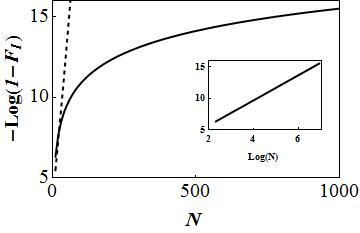}
\caption{Function of fidelity, $-\log(1-F(\rho_A,\rho'_{A}))$, as a function of $N$ for $\beta \rightarrow \infty$. The field strength is $h=0.9, 1.$ for the dashed and black lines.}\label{Fig10}
\end{figure}

\section{Conclusions}

In this work we studied the locality aspect of the Zeroth law of thermodynamics for quantum spin chains with strong but finite range interactions. Upon noting that in the presence of strong interactions the marginal states of a global thermal state do not take the canonical form themselves, we go on defining an effective thermal state for a subsystem. The latter being the reduced density matrix of the subsystem considered as a part of a slightly bigger, enveloping thermal system (see Fig.~\ref{fig:BoundaryRegion}). Borrowing concepts from quantum information theory and employing methods from quantum statistical mechanics, we relate the accuracy with which the effective thermal state describes the actual state of the subsystem to the correlations present in the whole system (see Eqs.~(\ref{dist2}, \ref{eq:def_cov}, \ref{eq:truncation_error_in_terms_of_cov}) and the discussion around them). We further utilize a Trotter approximation formula \cite{hastings,Molnar:2014} to build a tensor network representation of the corresponding states of the subsystem to provide upper bounds on the aforementioned accuracy, depending on the size of the enveloping thermal system, and such physical quantities as the spectral gap of the global hamiltonian and the temperature of the parent chain. At the quantum critical point, we use already existing asymptotical formulas from the conformal field theory.

Lastly, we exemplify our analytical findings by analyzing the quantum Ising chain. The latter is complex enough to have a quantum phase transition point, but simple enough to allow for an exact diagonalization by standard tools of statistical mechanics, thereby serving as a perfect testbed for our analytical upper bounds. In particular, we find that, e.g., away from criticality, the envelope which is bigger than the system only by one layer of spins, is enough to approximate the actual state with a rather high precision (see, e.g., Fig.~\ref{Fig3}).

Our results for one dimensional systems with finite range interactions suggest that investigating the properties of the effective thermal states in higher dimensions and, possibly, harbouring long range of interactions, is an interesting direction for further research, which can have far-reaching implications in efficient simulation of the subsystems of large and strongly interacting quantum systems.
Another interesting open question beyond the scope of this work 
is whether these results can be generalized to other other types of equilibrium states, e.~g.~
the so called \emph{Generalized Gibbs Ensemble} and steady states of local Liouvillians.

In a more practical vein, another field where our findings may find implications is quantum thermometry with non-negligible interactions \cite{antonella}.

\addcontentsline{toc}{section}{Acknowledgments}

\section*{Acknowledgments} This work is supported by the IP project SIQS, the Spanish project FOQUS, and the Generalitat de Catalunya (SGR 875). S.H.S. and M.P.L. acknowledge funding from the "la Caixa"-Severo Ochoa program, M.P.L is also supported by the Spanish grant FPU13/05988, A.R. is supported by the Beatriu de Pin\'os fellowship (BP-DGR 2013), L.T. is supported by the ERC AdG OSYRIS and A.A. is supported by the ERC CoG QITBOX. All authours thank the EU COST Action MP1209 ``Thermodynamics in the quantum regime''.

\appendix

\section{Proofs of the Lemmas}\label{appa2}
In this section, we present the proofs of the lemma's used in Secs. 3 and 4
to get statements on the locality of temperature for gapped systems.
They consist of how different covariances decay for one dimensional systems with a gapped transfer matrix $T$. 

\begin{lemma}\label{lemma:covarianceIS}[Infinite chain]
Given a gapped transfer matrix $T$ with eigenvalues $\lambda_k$ labelled
in decreasing order, i.~e.\ $|\lambda_k|\geq |\lambda_{k'}|$ for all $k<k'$, a right (left) dominant eigenvector $\ket{1_R}$ ($\bra{1_L}$),
the largest eigenvalue $\lambda_1=1$,
and a covariance between any two operators $O$ and $O'$ separated by a distance $\ell$ defined as
\begin{equation}
\cov(\ell;O,O',T)=\bra{1_L} O^\dagger T^{\ell} O' \ket{1_R}-
\bra{1_L} O^\dagger T^{\ell} \ket{1_R}\bra{1_L} T^{\ell} O' \ket{1_R}\, .
\end{equation}
Then, the covariance can be proven to
decay exponentially in $\ell$
\begin{equation}
|\cov(\ell;O,O',T)|\le \norm{ O\ket{1_L}}\norm{ O'\ket{1_R}}\e^{-\ell/\xi}
\end{equation}
where $\xi=-\left(\ln |\lambda_2|\right)^{-1}>0$ is the correlation length and
the $\norm{\ket{\varphi}}=\langle \varphi \ket{\varphi}^{1/2}$ is the norm of the vector $\ket{\varphi}$.
\end{lemma}
\begin{proof}
Let us first introduce the new transfer matrix $\tilde T= T-\ket{1_R}\bra{1_L}$ to rewrite the covariance as
\begin{equation}
\cov(\ell;O,O',T)=
\bra{1_L}  O^\dagger \tilde T^{\ell} O' \ket{1_R}\, ,
\end{equation}
where we have used that $\tilde T^\ell=T^\ell -\ket{1_R}\bra{1_L}$.
By using the Cauchy-Schwarz inequality one gets
\begin{equation}
\left|\bra{1_L} O^\dagger  \tilde T^{\ell} O' \ket{1_R}\right|\le
\norm{O'\ket{1_R}}\left(\bra{1_L}  O^\dagger   \tilde T^{\ell} (\tilde T^\dagger)^{\ell} O \ket{1_L}\right)^{1/2} .
\end{equation}
Let us now consider second factor separately.
By inserting a resolution of the identity, a straight forward calculation
leads to
\begin{eqnarray}\label{eq:normofaTvector}
\hspace{-1cm}
\bra{1_L} O^\dagger \tilde T^{\ell} (\tilde T^\dagger)^{\ell}  O \ket{1_L}&=&
\sum_{k\ge 2}|\lambda_k|^{2\ell}\bra{1_L}O^\dagger \ket{k_R}\bra{k_R} O\ket{1_L}
\nonumber \\
&\le & |\lambda_2|^{2\ell}\sum_{k\ge 2}\bra{1_L} O^\dagger \ket{k_R}\bra{k_R} O\ket{1_L} 
\le |\lambda_2|^{2\ell} \norm{ O \ket{1_L}}^2
\end{eqnarray}
where we have used that $|\lambda_2|$ is an upper-bound for all the
$|\lambda_k|$ with $k\ge 2$ and the Parseval inequality. 

Finally, we put everything together and get
\begin{equation}
|\cov(\ell;O,O',T)|\le \norm{ O\ket{1_L}}\norm{ O'\ket{1_R}} \e^{-\ell/\xi}
\end{equation}
where the 2nd largest eigenvalue $|\lambda_2|$ has been written in terms of the correlation length $\xi$.
\end{proof}

\begin{lemma}\label{lemma:covariancePBC}[Periodic boundary conditions] 
Given a system with periodic boundary conditions, an Hermitian transfer matrix T with a gap $\Delta$ 
and a covariance between any two operators $O$ and $O'$ separated by a distance $\ell$ defined as
\begin{equation}\label{eq:cov-matrix-form2}
\cov(\ell;n,O,O',T)=\frac{\Tr (O T^\ell O'T^{n-\ell}) }{\Tr (T^n)}-
\frac{\Tr(OT^n)}{\Tr (T^n)}\frac{\Tr(O'T^n)}{\Tr (T^n)} \, .
\end{equation}
where $0 \le \ell \le n$ and $n$ is the system size.
Then, the covariance $\cov(\ell)=\cov(\ell;n,O,O',T)$ as a function of $\ell$ fulfills following properties:
\begin{description}
\item[({\it i})] Its real part is symmetric respect to the $n/2$ and the interchange of $A$ and $B$, i.~e.
\begin{equation}\label{eq:symmetric}
\cov(n-\ell;n,O,O',T)=\cov(\ell;n,O,O',T)^*=\cov(\ell;n,O,O',T)\, .
\end{equation}
\item[({\it ii})]  Given two operators $O$ and $O'$, there always exist two other operators $O_M$ and $O_M'$ such that
\begin{equation}\label{eq:convex-bound}
|\cov(\ell;n,O,O',T)|\le \cov(\ell;n,O_M,O_M',T)
\end{equation}
and where $\cov(\ell;n,O_M,O_M',T)$ is a convex function in $\ell$ that is maximum at $\ell=0$ and 1, 
and reaches its minimum at $\ell=n/2$.
\end{description}
\end{lemma}
\begin{proof}
Statement ({\it i}) is a simple consequence of the following elementary equalities
\begin{eqnarray}
\left(\Tr (O T^\ell O'T^{n-\ell})\right)^*&=&\Tr\left( (O T^\ell O'T^{n-\ell})^\dagger\right) \nonumber \\
&=&\Tr(T^{n-\ell} O' T^\ell O) 
=\Tr(O' T^\ell O T^{n-\ell} ) \, .
\end{eqnarray}
In order to prove ({\it ii}), let us focus on
the first term in Eq.~(\ref{eq:cov-matrix-form2}),
since note that the second one does not depend on $\ell$. With this aim, we define
\begin{equation}
f(\ell; O,O')=\Re \left[ \frac{\Tr (O T^\ell O'T^{n-\ell}) }{\Tr (T^n)}\right]\, .
\end{equation}
By introducing the transfer matrix in its spectral representation, $f(\ell)$ can be written as
\begin{eqnarray}\label{eq:f_ell_spectral}
f(\ell; O,O')&=&\frac{s_1^n}{\Tr (T^n)}\left[\bra{1}O\proj{1}O'\ket{1}
+ \sum_{k\ge 2}c_k\left( \left(\frac{\lambda_k}{\lambda_1}\right)^\ell 
+\left(\frac{\lambda_k}{\lambda_1}\right)^{n-\ell}\right)\right. \nonumber \\
&+& \left. 
\sum_{k,k'\ge 2}d_{kk'}\left(\frac{\lambda_k}{\lambda_1}\right)^{\ell}\left(\frac{\lambda_{k'}}{\lambda_1}\right)^{n-\ell}\right]\, ,
\end{eqnarray}
where $c_k=\Re \left(\bra{1}O\proj{k}O' \ket{1}\right)$ and $d_{kk'}=\Re\left(\bra{k}O\proj{k'}O'\ket{k}\right)$.
Note now that 
\begin{equation}
\left(\frac{\lambda_k}{\lambda_1}\right)^\ell+\left(\frac{\lambda_k}{\lambda_1}\right)^{n-\ell}
=2\e^{-\frac{n}{2\xi_k}}\cosh\left( \frac{\ell-n/2}{\xi_k}\right) \, ,
\end{equation}
where the correlation length $\xi_k$ is defined as
\begin{equation}
\xi_k^{-1}:= \ln \left(\frac{\lambda_1}{\lambda_k} \right) \, .
\end{equation} 
Note that as the eigenvalues of the transfer matrix are ordered, a larger $k$ implies a shorter correlation length $\xi_k$.

In a similar way, we can also simplify the terms in the last sum in Eq.~(\ref{eq:f_ell_spectral}). Note that
\begin{eqnarray}
\left(\frac{\lambda_k}{\lambda_1}\right)^\ell\left(\frac{\lambda_{k'}}{\lambda_1}\right)^{n-\ell}+
\left(\frac{\lambda_{k'}}{\lambda_1}\right)^\ell\left(\frac{\lambda_{k}}{\lambda_1}\right)^{n-\ell}
&=&\e^{-\frac{n}{\xi_{k'}}-\frac{\ell}{\xi_{kk'}}}+\e^{-\frac{n}{\xi_{k}}+\frac{\ell}{\xi_{kk'}}}\nonumber\\
&=&2\e^{-\frac{n}{2\xi_{kk'}}} \cosh \left(\frac{\ell-n/2}{\xi_{kk'}}\right)\, .
\end{eqnarray}
where the length $\xi_{kk'}$ has been defined as $\xi_{kk'}^{-1}=\xi_k^{-1}-\xi_{k'}^{-1}$.
Puting the previous steps together, we get
\begin{eqnarray}\label{eq:expression-f_ell}
f(\ell;A,B)&=&\frac{\lambda_1^n}{\Tr (T^n)}\left[\bra{1}A\proj{1}B\ket{1}+ 
2\sum_{k\ge 2}\e^{-\frac{n}{2\xi_k}}c_k\cosh\left( \frac{\ell-n/2}{\xi_k}\right) \right. \nonumber \\
&+& \left.2\sum_{k\ge 2} d_{kk}\e^{-\frac{n}{\xi_k}} + 2\sum_{2 \le k < k'}d_{kk'}\e^{-\frac{n}{2\xi_{kk'}}} \cosh \left(\frac{\ell-n/2}{\xi_{kk'}}\right)\right]\, .
\end{eqnarray}
Note that in general the covariance could oscillate in $\ell$, since $c_k$ and $d_{kk'}$ could take negative values for some $k$ and $k'$. 
Nevertheless, given two operators $A$ and $B$ for which some $c_k$ and $d_{kk'}$ are negative, there always exist two operators $O_M$ and $O_M'$ such that their respective $\tilde c_k=|c_k|$ and $\tilde d_{kk'}=|d_{kk'}|$. For instance,
$\bra{k}O_M\ket{k'}=-\bra{k}O\ket{k'}$ for the $k'$ and $k$-s with negative coefficents and $\bra{k}O_M\ket{k'}=\bra{k}O_M\ket{k'}$ otherwise, and $O_M' = O'$.
This covariance $f(\ell; O_M,O_M')$ is an upper bound to the absolute value of the previous one covariance $f(\ell; O,O')$,
\begin{eqnarray}\label{eq:bound-f_ell}
f(\ell;O_M,O_M') &=& \frac{\lambda_1^n}{\Tr (T^n)}\left[|\bra{1}O\proj{1}O'\ket{1}|+ 
2\sum_{k\ge 2}\e^{-\frac{n}{2\xi_k}}|c_k|\cosh\left( \frac{\ell-n/2}{\xi_k}\right) \right. \nonumber \\
\hspace{-1cm}&+& \left.2\sum_{k\ge 2} |d_{kk}|\e^{-\frac{n}{\xi_k}}+2\sum_{2 \le k < k'}|d_{kk'}|\e^{-\frac{n}{2\xi_{kk'}}} \cosh \left(\frac{\ell-n/2}{\xi_{kk'}}\right)\right]
\nonumber \\
\hspace{-1cm}&\ge&  f(\ell; O,O')\, .
\end{eqnarray}
As the sum in Eq.~(\ref{eq:bound-f_ell}) is a linear combination of convex functions with positive coefficents,
 $f(\ell;O_M,O_M')$ is also convex.
It is also obvious from the properties of the $\cosh()$ function, that $f(\ell;O_M,O_M')$ reaches its maximum at $\ell=0$ and $n$, and its minimum at $\ell=n/2$.
\end{proof}





\section{Some bounds on the generalized covariance} \label{appfin}

In this appendix we bound the different terms in Eq.~(\ref{eq:GeCov-as-2PCF}) to provide an upper-bound 
for the generalized covariance used in the Section~\ref{section:finite-temperature}.

Let us start with the first term in Eq.~(\ref{eq:GeCov-as-2PCF}). By using Lemma \ref{lemma:covarianceIS}, we get
\begin{equation}
|\cov_T(\ell_B;X_s,YT^{\ell_B}T_s)| \le \norm{X_s \ket{1_L}} \norm{YT^\ell_B T_s\ket{1_R}} \e^{-\ell_B/\xi}\, .
\end{equation}
The coefficients can be bounded by
\begin{eqnarray}
 \norm{X_s \ket{1_L}}^2=
\bra{1_L}X_sX_s^\dagger \ket{1_L}\le \norminf{h_L}^2 \bra{1_L}T_sT_s^\dagger \ket{1_L}
  \end{eqnarray}
\begin{eqnarray}
  \norm{Y T^{\ell_B} T_s \ket{1_R}}^2&\le& \norminf{O}^2 \norm{T^{\ell_B+1}T_s\ket{1_R}}^2\nonumber \\
  &\le&  \norminf{O}^2 \left(|\av{T_s}_T|^2+\bra{1_R}T_s^\dagger T_s \ket{1_R}\e^{-2(\ell+1)/\xi}\right)\, ,
\end{eqnarray}
where we have used the structure of $X_s$ and $Y$ and 
the fact that $\Tr(O\rho)\le \norminf{O}\Tr(\rho)$ for any $\rho\ge 0$.
It is also necessary to upper bound the inverse of the partition function.
With that aim, let us write the partition function as
\begin{equation}
Z_s=\av{T_s}_T^2+\bra{1_L}T_s\ \left(T-\ket{1_R}\bra{1_L}\right)^{2\ell_B+1}\,T_s \ket{1_R}\, 
\end{equation}
where we have used that $(T -\ket{1_R}\bra{1_L})^\ell=T^\ell -\ket{1_R}\bra{1_L}$.
By proceeding similarly as in the proof of Lemma~\ref{lemma:covarianceIS}, 
the partition function $Z_s$ can be lower bounded by
\begin{equation}
|Z_s|\ge|\av{T_s}_T|^2-\norm{T_s^\dagger \ket{1_L}}\norm{T_s \ket{1_R}}\e^{-(2\ell_B+1)/\xi}\, ,
\end{equation}
and its inverse is upper bounded by
\begin{eqnarray}
|Z_s|^{-1}&\le&|\av{T_s}_T|^{-2}\left(1 + \frac{\norm{T_s^\dagger \ket{1_L}}\norm{T_s \ket{1_R}}}{|\av{T_s}_T|^4}\e^{-(2\ell_B+1)/\xi}\right.\nonumber\\
 &+& \left. \frac{\norm{T_s^\dagger \ket{1_L}}^2\norm{T_s \ket{1_R}}^2}{|\av{T_s}_T|^{8}}\e^{-(4\ell_B+2)/\xi}\right)\, .
\end{eqnarray}
By putting the previous bounds together we get
\begin{equation}
\left|\frac{\cov_T(\ell_B;X_s,YT^{\ell_B}T_s)}{Z_s}\right| \le
\norminf{O} \norminf{h_L}(1+\sigma_L(s))
 \e^{-\ell_B/\xi} + \mathcal{O}(\e^{-2\ell_B / \xi}) ,
\end{equation}
where $\sigma_{L}(s)=(\bra{1_{L}}T_sT_s^\dagger\ket{1_{L}}-|\av{T_s}_T|^2)^{1/2}/|\av{T_s}_T|$
is the relative standard deviation of $T_s$ on the left dominant eigenstate $\ket{1_L}$
and we have omitted for simplicity the second order terms in $\e^{-\ell_B/\xi}$.

In a similar way, the second term in Eq.~(\ref{eq:GeCov-as-2PCF}) can be bounded by
\begin{eqnarray}
\left|\frac{\av{X_s}_T\av{T_s}_T}{Z_s}\frac{\cov_T(\ell_B;T_s,YT^{\ell_B}T_s)}{Z_s}\right|&\le&
\norminf{h_L}\norminf{O}(1+\sigma_L(s))\e^{-\ell_B/\xi} \nonumber\\
&+& \mathcal{O}(\e^{-2\ell_B / \xi}) ,
\end{eqnarray}
where we have used that $ |\av{X_s}_T/\av{T_s}_T|\le \norminf{h_L}$.

The rest of terms in Eq.~(\ref{eq:GeCov-as-2PCF}) can be analogously bounded. Note that they will only contribute to the second order. Putting everything together in Eq.~(\ref{eq:trivial-bound}), the physical distinguishability on the region $A$ between the truncated and untruncated thermal states is upper-bounded by
\begin{equation}
\hspace{-.5cm}
    \Tr \left[O\,\omega(H_{AB})\right] 
	 - \Tr\left[O\,\omega(H)\right]
    \le 2 \beta c \norminf{h_L}\norminf{O} e^{-\ell_B /\xi} +\mathcal{O}(\e^{-2\ell_B / \xi}), 
\end{equation}
where $c=1+\int_0^1 \dd s \sigma_L(s)$ is a constant of order one, depending on the model.

\section{Solving quantum Ising model} \label{appa}

In this appendix we find the states (\ref{eq:rhoA}) and (\ref{eq:rhoA'}) using formalism of covariance matrices.


\subsection*{\textbf{Jordan-Wigner transformation}}
Let us first apply  the \textit{Jordan-Wigner transformation}, $\sigma_x^i \otimes \sigma_x^{i+1}=(a_i^\dagger-a_i)(a_{i+1}+a_{i+1}^\dagger)$ and $\sigma_z^i=a_i a_i^\dagger-a_i^\dagger a_i$, to the Hamiltonian (\ref{eq:numerical-example}). We obtain,
\begin{equation}\label{eq02}
H_n=\sum_{i,j=1}^N A_{ij} a_i a_j^\dagger+\frac{1}{2}\sum_{i,j=1}^N B_{ij}(a_i^\dagger a_j^\dagger-a_i a_j) \text{,}
\end{equation}
with $A_{ij}=h \delta_{i,j} + \frac{1}{2}(\delta_{i+1,j}+\delta_{i,j+1})$ and $B_{ij}= \frac{1}{2}(\delta_{i+1,j}-\delta_{i,j+1})$ and where $a_i$ and $a_i^\dagger$ denote annihilation and creation operators, respectively. From this form of the Hamiltonian, we notice it is quadratic, and thus the thermal state (and their marginal states) are gaussian states. Therefore we can deal with them using the covariance matrix formalism.

\subsection*{\textbf{The correlation matrix}}
In this formalism, we define the global correlation matrix, $\Gamma$, as

\begin{equation}\label{eq05}
\hspace{-2.2cm}\Gamma (X) = \langle X X^\dagger \rangle =
\left[\begin{array}{ccc}
\parallel\langle a_i a_j ^\dagger\rangle \parallel_{N \times N} &  & \parallel\langle a_i a_j\rangle\parallel_{N \times N} \\
 & & \\
\parallel\langle a_i^\dagger a_j ^\dagger\rangle\parallel_{N \times N} & & \parallel\langle a_i^\dagger a_j\rangle\parallel_{N \times N}
\end{array} \right] 
\quad \text{with} \quad
X = \left[ \begin{array}{c}
a_1 \\ 
\vdots \\
a_{N} \\ 
a_1^{\dagger} \\
\vdots \\
a_{N}^\dagger
\end{array} \right] \text{,}
\end{equation}
where $\parallel ... \parallel_{N \times N}$ refers to a $N \times N$ matrix. Given $\Gamma$, we can obtain the correlation matrix corresponding to a reduced state by just selecting the corresponding matrix elements of $\Gamma$. For example, the correlation matrix of the fermions $k,k+1$ is given by,
\begin{equation}\label{eq03}
\Gamma_{k,k+1}=\left( \begin{array}{cccc}
\langle a_k a_k^\dagger \rangle & \langle a_k a_{k+1}^\dagger \rangle & \langle a_k a_k\rangle & \langle a_k a_{k+1} \rangle \\
\langle a_{k+1} a_k^\dagger \rangle & \langle a_{k+1} a_{k+1}^\dagger \rangle & \langle a_{k+1} a_k \rangle & \langle a_{k+1} a_{k+1} \rangle \\
\langle a_k^\dagger a_k^\dagger \rangle & \langle a_k^\dagger a_{k+1}^\dagger \rangle & \langle a_k^\dagger a_k\rangle & \langle a_k^\dagger a_{k+1} \rangle \\
\langle a_{k+1}^\dagger a_k^\dagger \rangle & \langle a_{k+1}^\dagger a_{k+1}^\dagger \rangle & \langle a_{k+1}^\dagger a_k \rangle & \langle a_{k+1}^\dagger a_{k+1} \rangle
\end{array} \right)\text{.}
\end{equation}

Since the Jordan Wigner transformation is local, in the sense that it maps the $k$th fermion to the $k$th spin in the chain, this correlation matrix also corresponds to the two-spin subsystem at sites $k$ and $k+1$. This subsystem is precisely the region of interest $A$ in section \ref{localityIsing}, and thus (\ref{eq03}) corresponds to the correlation matrix of $\rho_A$ in (\ref{eq:rhoA}).

Given the reduced correlation matrix, the explicit form of $\rho_A$ can be easily obtained. As the reduced state of a thermal state is gaussian, there is a one to one connection between (\ref{eq03}) and $\rho_A$. Indeed, for any gaussian state, with
\begin{equation}
\rho = \frac{e^{-X^\dagger M X}}{Tr[e^{-X^\dagger M X}]}\text{ with $M$ a coefficient matrix,}
\end{equation}
it is straightforward to prove that, provided that $M$ is diagonalizable, 
\begin{equation}
\Gamma (X) = \frac{1}{(\mathbb{1}+e^{-2M})}
\end{equation}
or, equivalently, that
\begin{equation}\label{Mrho}
M=-\frac{1}{2} \log(\Gamma(X)^{-1}-\mathbb{1}) \text{.}
\end{equation}

\subsection*{\textbf{Explicit computation}}

Now we explicitly compute (\ref{eq03}) for a finite and an infinite chain, in order to obtain $\rho'_A$ and $ \rho_A$, respectively, using relation (\ref{Mrho}). \\

\begin{itemize}

\item Finite chain
\hfill \break
\hfill \break
For the case of a finite chain, we need to obtain the correlation matrix (\ref{eq05}) corresponding to the global state. It is then useful to first diagonalize the Hamiltonian (\ref{eq02}) by applying the Bogoliubov transformation 
\begin{equation}\label{eq06}
b_j=\sum_{k=1}^N \frac{1}{2} (\phi_{jk}+\psi_{jk})a_i -\frac{1}{2}(\phi_{jk}-\psi_{jk})a_i^\dagger \text{,}
\end{equation}
where $\phi$ and $\psi$ are real matrices and verify $\sum_{k=1}^N \phi_{jk}^2=\sum_{k=1}^N \psi_{jk}^2=1$. The Hamiltonian can then take the form, 
\begin{equation}
H=\sum_{k=1}^N \xi_k (b_k^\dagger b_k -1/2)\text{,}
\end{equation}
where $\xi_k$ are the fermionic excitation energies and $b_k$ and $b_k^\dagger$ denote annihilation and creation operators, respectively. The excitation energies, $\xi_k$, and the matrices $\phi$ and $\psi$ are obtained by solving the equation
\begin{equation}
(A-B) \phi = \psi D \text{,}
\end{equation}
where $D$ is a diagonal matrix whose entries correspond to the excitation energies, $\xi_k$.

Once the Hamiltonian is diagonalized, it is easy to compute the correlation matrix of a thermal state at inverse temperature $\beta$ in the diagonalized basis, obtaining


\begin{equation}\label{eq05}
\Gamma (Y) =
\left[\begin{array}{ccc}
\frac{1}{1+e^{-\beta D}} &  & 0_{N \times N} \\
 & & \\
0_{N \times N} & & \frac{1}{1+e^{\beta D}}
\end{array} \right] 
\quad \text{with} \quad
Y = \left[ \begin{array}{c}
b_1 \\ 
\vdots \\
b_{N} \\ 
b_1^{\dagger} \\
\vdots \\
b_{N}^\dagger
\end{array} \right] \text{,}
\end{equation}
where the non-zero matrices are diagonal.\\

From that expression we can obtain the correlation matrix in the original basis, $\Gamma (X)$, via 
\begin{equation}
\Gamma (X) = T^\dagger \Gamma (Y) T \text{,}
\end{equation}
where $T$ is the transformation matrix defined by the Bogoliubov transformation (\ref{eq06}). That is, $Y = T X$, with
\begin{equation}
T = \left[  \begin{array}{cc}
\gamma & \mu \\
\mu^* & \gamma^*
\end{array} \right] \text{.}
\end{equation}
and 
\begin{equation}
\gamma = \frac{1}{2} (\phi + \psi) \text{ and } \mu = - \frac{1}{2} (\phi - \psi) \text{.}
\end{equation}
\ \\
\item Infinite chain ($N \rightarrow \infty$)
\hfill \break
\hfill \break
In the case of an infinite chain, (\ref{eq03}) can be obtained relying on the analytical results from \cite{Barouch}. The partial state of a two-spin subsystem  is
\begin{equation}
\rho_2^{n\rightarrow\infty} =\frac{1}{4}\left[ 1+\langle \sigma_z^k \rangle (\sigma_z^k+\sigma_z^{k+1})+\sum_{l=x,y,z}^3 \langle \sigma_l ^k \sigma_l ^{k+1}\rangle \sigma_l ^k \otimes \sigma_l ^{k+1} \right] \text{,}
\end{equation}
where the average $\langle\sigma_z^k\rangle$ and the two-spin correlation functions $\lbrace\langle\sigma_l ^k \sigma_l ^{k+1}\rangle\rbrace_{l=\lbrace x,y,z\rbrace}$ are given by \cite{Barouch}. In order to express the state in the fermionic basis, we can compute the reduced correlation matrix (\ref{eq03}) from this state,
\begin{equation}
\Gamma_{k,k+1} = \left( \begin{array}{cccc}
2(1+\alpha)& -(\beta+\gamma) & 0 & -(\beta-\gamma) \\
-(\beta+\gamma)& 2(1+\alpha) & \beta-\gamma & 0 \\
0 & \beta-\gamma & 2(1-\alpha) & \beta+\gamma \\
-(\beta-\gamma) & 0 & \beta+\gamma & 2(1-\alpha)
\end{array} \right) \text{,}
\end{equation}
with $\alpha= \langle \sigma_z^k\rangle $, $\beta= \langle \sigma_x^k \sigma_x^{k+1} \rangle $ and $\gamma= \langle \sigma_y^k  \sigma_y^{k+1}\rangle $.


\end{itemize}


\addcontentsline{toc}{section}{References}

\section*{References}
\begin{thebibliography}{99}%

\bibitem{bloch} I.~Bloch, J.~Dalibard, and S.~Nascimb\`ene, Nat. Phys. {\bf 8}, 267 (2012).

\bibitem{polkan} A.~Polkovnikov, K.~Sengupta, A.~Silva, and M.~Vengalattore, Rev. Mod. Phys. {\bf 83}, 863 (2011).

\bibitem{plenio} J.~Kai, A.~Retzker, F.~Jelezko, and M.B.~Plenio, Nat. Phys. {\bf 9}, 168 (2013).

\bibitem{greiner} J.~Simon, {\it et al}, Nature {\bf 472}, 307 (2011).

\bibitem{vedral} K.~Maruyama, F.~Nori, and V.~Vedral, Rev. Mod. Phys. {\bf 81}, 1 (2009).

\bibitem{popka} S.~Popescu, A.J.~Short, and A.~Winter, Nat. Phys. {\bf 2}, 754 (2006). 

\bibitem{oppen} M.~Horodecki, and J.~Oppenheim, Nat. Commun. {\bf 4}, 2059 (2013).

\bibitem{lidia} L.~del~Rio, {\it et al}, Nature {\bf 474}, 61 (2011).

\bibitem{tumulka} S.~Goldstein, J.L.~Lebowitz, R.~Tumulka, and N.~Zanghi, Phys. Rev. Lett. {\bf 96}, 050403 (2006).

\bibitem{Gemmer} J.~Gemmer, M.~Michel, and G.~Mahler, \emph{Quantum Thermodynamics}, (Berlin: Springer, 2004).

\bibitem{short} A.J.~Short, and T.C.~Farrelly, New J. Phys. {\bf 14}, 013063 (2012).

\bibitem{brandao} F.G.S.L.~Brand\~ao, and M.~Cramer, \emph{Equivalence of Statistical Mechanical Ensembles for Non-Critical Quantum Systems}, arXiv:1502.03263 [quant-ph].

\bibitem{reichl} L.E.~Reichl, \emph{A Modern Course in Statistical Physics} (New York: Wiley, 1998).

\bibitem{balian} R.~Balian, \emph{From Microphysics to Macrophysics}, (Heidelberg: Springer, 2007).

\bibitem{ferraro} A.~Ferraro, A.~Garc\'ia-Saez, and A.~Ac\'in, 
{Europhys. Lett. {\bf 98}, 10009 (2012).}

\bibitem{artur} A.~Garc\'ia-Saez, A.~Ferraro, and A.~Ac\'in, Phys. Rev. A \textbf{79}, 052340 (2009).

\bibitem{Kliesch:2014} M.~Kliesch, C.~Gogolin, M.J.~Kastoryano, A.~Riera, and J.~Eisert, 
{Phys. Rev. X {\bf 4}, 031019 (2014)}.

\bibitem{pachon} L.A.~Pachon, J.F.~Triana, D.~Zueco, and P.~Brumer, \emph{Uncertainty Principle Consequences at Thermal Equilibrium}, arXiv:1401.1418 [quant-ph].

\bibitem{mikeike} M.A.~Nielsen, and I.L.~Chuang, {\it Quantum Computation and Quantum Information} (Cambridge: Cambridge University Press, 2000).

\bibitem{FvdG} C.A. Fuchs, and J. van de Graaf, {\it Trans IEEE. Inf. Theory} {\bf 45}, 1216–27 (1999).

\bibitem{DiFrancesco:1997} P.~Di Francesco, P.~Mathieu, and D.~Senechal, \emph{Conformal Field Theory} (New York: Springer, 1997).

\bibitem{Cardy:2008} J.~Cardy, \emph{Boundary conformal field theory}, 
\emph{Encyclopedia of mathematical Physics ed J-P Françoise,GL Naber and T S Tsun}
(Amsterdam: Elsevier).

\bibitem{lluis} M.P.~Mueller, E.~Adlam, L.~Masanes, and N.~Wiebe, \emph{Thermalization and canonical typicality in translation-invariant quantum lattice systems}, arXiv:1312.7420 [quant-ph].

\bibitem{fss} J.~Cardy, \emph{Finite-size scaling} (Amsterdam: Elsevier, 2012).

\bibitem{Gelfert2001}
A.~Gelfert, and W.~Nolting, J. Phys.: Condens. Matter {\bf 13}  505-24 (2001).

\bibitem{hastings} M.B.~Hastings, Phys. Rev. B {\bf 73}, 085115 (2006).
	 
\bibitem{Molnar:2014} A.~Moln\'ar, N.~Schuch, F.~Verstraete, and J.I.~Cirac, 
{Phys. Rev. B {\bf 91}, 045138 (2014)}.

\bibitem{Kinoshita} T. Kinoshita, T. Wenger, and D.S. Weiss, Nature \textbf{440}, 900 (2006).

\bibitem{Hofferberth} S. Hofferberth et al., Nature \textbf{449}, 324 (2007).

\bibitem{Bloch} I. Bloch, J. Dalibard, and W. Zwerger, Rev. Mod. Phys. \textbf{80}, 885 (2008).

\bibitem{Rigol} M. Rigol, V. Dunjko, and M. Olshanii, Nature \textbf{452}, 854 (2008).

\bibitem{Banuls} M.C.~Ba\~nuls, J.I.~Cirac, and M.B.~Hastings, Phys. Rev. Lett. \textbf{106}, 050405 (2011).

\bibitem{Linden} N. Linden, S. Popescu, and P. Skrzypczyk, Phys. Rev. Lett. \textbf{105}, 130401 (2010).

\bibitem{Gogolin} C. Gogolin, M.P. Mueller, and J. Eisert, Phys. Rev. Lett. \textbf{10}, 040401 (2011).

\bibitem{Osborne} T.J. Osborne, and M.A. Nielsen, Phys. Rev. A \textbf{66}, 032110 (2002).

\bibitem{Barouch} E. Barouch, and B.M. McCoy, Phys. Rev. A \textbf{3}, 786 (1970).

\bibitem{Nielsen2} M.A. Nielsen, \textit{The Fermionic canonical commutation relations and the Jordan-Wigner transform}, School of Physical Sciences (The University of Queensland, 2005).

\bibitem{Cardy1984}
J. Cardy.
\newblock {\em Nucl. Phys. B} \textbf{240}, 514-32 (1984).

\bibitem{Cardy1986}
J. Cardy.
\newblock {\em Nucl. Phys. B} \textbf{275}, 200-18 (1986).

\bibitem{Perez-Garcia2007}
D.~Perez-Garcia, F.~Verstraete, M.~M. Wolf, and J.~I. Cirac.
\newblock {\em Quantum Info. Comput.} \textbf{7} , 401-30 (2007).

\bibitem{Pirvu2012a}
B.~Pirvu, G.~Vidal, F.~Verstraete, and L.~Tagliacozzo.
\newblock {\em Phys. Rev. B}, \textbf{86} 075117 (2012).

\bibitem{Verstraete2006}
F.~Verstraete, and J.~I. Cirac.
\newblock {\em Phys. Rev. B}, \textbf{73} 094423 (2006).

\bibitem{Michael} M. Lubasch, J.I. Cirac, M.C. Ba\~nuls.
\newblock {\em 	Phys. Rev. B} \textbf{90}, 064425 (2014).

\bibitem{MichaelII} M. Lubasch, J.I. Cirac, M.C. Ba\~nuls.
\newblock {\em New J. Phys.}  \textbf{16}, 033014 (2014)

\bibitem{antonella} A. De Pasquale, D. Rossini, R. Fazio, V. Giovannetti.
\newblock \emph{Local quantum thermometry}
\newblock arXiv:1504.07787 [quant-ph].

\end{thebibliography}
\end{document}